\documentclass[twoside,leqno,twocolumn]{article}  
\usepackage{balance}
\usepackage{amsmath, amssymb, amsthm}
\usepackage{color}
\usepackage{graphicx,url}
\usepackage{enumerate}

\newcommand{\rmconf}[1]{#1}
\newcommand{\addconf}[1]{}

\newtheorem{theorem}{Theorem}
\newtheorem{lemma}{Lemma}

\newtheorem{corollary}{Corollary}[theorem]
\newtheorem{corollarylemma}{Corollary}[lemma]
\newtheorem{property}{Property}

\def\R{\mathbb{R}}

\newcommand{\comment}[1]{	}

\definecolor{light-gray}{gray}{0.4}

\newcommand{\example}[3]{\noindent{\bf Example #1.} ({\it #2}) #3 \qed}

\def\R{\mathbb{R}}

\newcommand{\INDEP}{{\cal I\,}}

\newcommand{\LP}{{\cal L\,}}
\newcommand{\OPT}{{\rm OPT\,}}
\newcommand{\OO}{{\rm O\,}}

\newcommand{\Prob}{\mathbb{P}}
\newcommand{\Ex}{\mathbb{E}}

\newcommand{\beq}{\begin{equation}}
\newcommand{\eeq}{\end{equation}}
\newcommand{\E}{{\mathbb E}}

\title{Correlation Robust Stochastic Optimization}
\author{
Shipra Agrawal\thanks
 { Email: {shipra@cs.stanford.edu}. Computer Science and Engineering, Stanford University, Stanford, CA 94305, USA. Research supported in part by Boeing.
  }
  \and
Yichuan Ding\thanks
  {Email: {y7ding@stanford.edu}. Management Science and Engineering, Stanford
University, Stanford, CA 94305, USA.}
\and
Amin Saberi \thanks{Email: {saberi@stanford.edu}. Management Science and Engineering, Stanford
University, Stanford, CA 94305, USA.} \and
Yinyu Ye \thanks{Email: {yinyu-ye@stanford.edu}. Management Science and Engineering
and, by courtesy, Electrical Engineering, Stanford, CA 94305, USA. Research supported in part by Boeing.}
}
\begin{document}
\newcommand{\shcomment}[1]{\textcolor{red}{\it TODO: #1}\\}
\date{}
\maketitle

\begin{abstract}
We consider a robust model proposed by Scarf, 1958, for stochastic optimization when
only the marginal probabilities of (binary) random variables are given, and the correlation between the random variables is unknown. In the robust model, the objective is to minimize expected cost against worst
possible joint distribution with those marginals. 
We introduce the concept of {\it correlation gap} to compare this model to the stochastic optimization model that ignores correlations and minimizes expected cost under independent Bernoulli distribution. We identify a class of functions, using concepts of summable cost sharing schemes from game theory, for which the correlation gap is well-bounded and the robust model can be approximated closely by the independent distribution model. As a result, we derive efficient approximation factors for many popular cost functions, like submodular functions, facility location, and Steiner tree. As a byproduct, our analysis also yields some new results in the areas of
social welfare maximization and existence of Walrasian equilibria, which may be of independent interest.
\end{abstract}

\section{Introduction}
Stochastic optimization models decision making under uncertain or unknown
problem data.
We consider stochastic optimization problems in which the uncertain
variable is the ``demand" set. For example, in stochastic network design
problems, the random variable is the subset of source-destination
pairs to be connected; in stochastic facility location problem, the
random variable is the subset of potential clients that will have a demand; and in stochastic set cover problem, it is the subset of
elements that need to be covered. In general, such a stochastic
program can be expressed as \beq\label{SP}
\begin{array}{ll}
\min_{x \in C} \Ex[f(x,S)],
\end{array}
\eeq where $x$ is the decision variable which lies in a constrained
set $C$, and the random subset $S \subseteq V$ cannot be observed
before the decisions $x$ is made. $f(x,S)$ is the \emph{cost
function} which depends on both the decision $x$ and the outcome
scenario $S$.
The objective of stochastic programming is to minimize the
expected cost, which depends on  the joint distribution of items in
$V$. \newline
\\
In stochastic optimization, it is typically assumed that the distribution of random variable is either known or can be sampled from \cite{shapiro02, charikar05, swamy}. In this model, sample average approximation (SAA) has been used give approximation algorithms for many two-stage stochastic discrete optimization problems, including stochastic set cover \cite{swamy}, uncapacitated facility location \cite{swamy}, and Steiner tree problem \cite{boosted}. Those models are suitable when one does have access to a lot of time invariant reliable statistical information. In this paper, we study the problem when information about a part of the distribution (marginals) is known.
In the case when only marginal probabilities $p_i$ of each element are available, a common heuristic is to assume that the distribution of random set $S$ a product distribution. In other words, each
element $i$ may appear in $S$ independently with a given probability
$p_i$. For example, see
\cite{Kleinberg97allocatingbandwidth, 331530}. However, there is a
conventional wisdom that ignoring correlations can have catastrophic
consequences. Examples can be constructed such that the cost of the
solution optimized against the independent distribution performs
very poorly once certain correlations are introduced.\newline
\\
To address such problems, Scarf (1958, \cite{scarf}) proposed a
correlation-robust or distributionally-robust stochastic model,
which minimizes the expected cost over distributions having a fixed
marginal probability $p_i$ for each $i \in V$, but with any possible
correlations. For a problem instance $(f, V, \{p_i\})$, we wish to
find

\begin{equation}
\label{DRSPB}
\begin{array}{lrl}
 & \min_{x \in C} & g(x),
\end{array}
\end{equation}
where $g(x)$ is the expected cost under worst-case distribution when decision $x$ has been made, given by
\begin{equation}
\label{primal}
\begin{array}{lcl}
\max_{\cal D} &  \Ex_{\cal D}[f(x,S)] & \\
s.t. & \sum_{S: i\in S} \Prob_{\cal D}(S) =p_i. & \forall i\in V.
\end{array}
\end{equation}
We believe this is a very useful model because it takes advantage of
the stochasticity of the input, and at the same time efficiently
utilizes the available information. On the other hand, it defines an
exponential size linear program which makes the problem potentially
difficult to solve. A common strategy for such linear programs is to solve the corresponding dual LP with exponential number of constraints, using separating hyperplane approach. However, for the above model, approximating the separating hyperplane problem can be shown to be harder than the max-cut problem even for the special case when the function $f$ is submodular in $S$.\newline
\\
A natural question is how much risk it involves to simply ignore the correlations and minimize the expected cost of independent distribution instead of the worst case distribution. Or, in other words, how well the stochastic optimization model with independent distribution approximates the correlation robust model. The focus of this paper is to study this \emph{correlation gap}. For a particular problem instance $(f,V,\{p_i\})$ and a decision $x$, we define the correlation gap as the ratio between the expected cost $\Ex[f(x,S)]$ under the worst case distribution and that under the independent distribution on $S$. 
Correlation gap has many interesting implications for stochastic optimization problems.
A small upper bound on correlation gap allows relaxation of the stochastic optimization problem under any distribution, including the worst case distribution model \eqref{DRSPB}, to the product distribution case which is often more efficient to solve either by sampling or by other algorithmic techniques \cite{Kleinberg97allocatingbandwidth, 331530}. Further, in many real data collection scenarios, practical constraints can make it very difficult (or costly) to learn the complete information about correlations in data. In those cases, the correlation gap can provide a guideline to decide how important it is to spend resources on learning these correlations. In other words, it measures the ``value of correlations" in the statistical data.
Our main result is to characterize a wide class of functions for which the correlation gap can be well bounded. We also provide counter-examples showing large correlation gap for various other classes of functions.\newline
\\
Below, we summarize our key results:

\begin{itemize}
\item {\it A class of functions with bounded correlation gap:} For functions $f(x,S)$ that are non-decreasing in $S$ and have a cross-monotone, $\beta$-budget balance, (weak) $\eta$-summable cost-sharing scheme,  we show that the correlation gap is upper bounded by $\eta\beta \frac{e}{e-1}$. This will give correlation gap bounds (and matching approximation factors for robust model) of  $e/(e-1)$ for submodular functions, $\OO(\log{n})$ for facility location, and $\OO(\log^2 n)$ for Steiner forest, where $n=|V|$, the size of ground set.

\item {\it Hardness results:} We show examples with correlation gap of $\Omega(2^n)$ for functions supermodular in $S$, $\Omega(\sqrt{n}\log\log{n}/\log{n})$ for monotone subadditive functions in $S$, and $e/(e-1)$ for submodular functions. These examples will also prove corresponding lower bounds on approximation factors that can be achieved by substituting independent distribution for the robust model.  

\item{\it Polynomial-time algorithm for supermodular functions:} 
We analytically characterize the worst case distribution when function $f(x,S)$ is supermodular in $S$, and consequently give a polynomial-time algorithm for the correlation robust model provided $f$ is convex in $x$.

\item {\it New results for welfare maximization problems:} 
As a byproduct, our result provides a $\frac{1}{\eta\beta}(1-1/e)$-approximation algorithm for the well-studied problem of social welfare maximization in combinatorial auctions, when the utility functions are identical and admit $(\eta, \beta)$-cost-sharing scheme. Notably, this implies $(1-1/e)$-approximation for {\it identical} submodular utility functions, matching the best approximation factor (Vondrak, 2008 \cite{vondrak08}) for this case.

We also provide a simple counterexample for the conjecture by Bikhchandani \cite{bikhchandani97} that markets that have buyers with identical submodular utilities admit a Walrasian price equilibria.
\end{itemize}
The rest of the paper is organized as follows. To begin,
Section 2 will provide a mathematical definition of correlation gap, and
examples showing large correlation gap for certain classes of cost functions.
In Section 3, we present our main technical theorem that upper bounds the correlation gap for 
a wide class of cost functions, and discuss its implications on various stochastic optimization problems and the welfare maximization problem. The proof of this theorem is presented in Section 4. Finally, in Section 5, we end with a direct solution of correlation robust model for supermodular functions.
\section{Correlation Gap}
\label{sec:cor}
For a problem instance $(f,V,\{p_i\})$ and at a given decision $x$, we define correlation gap as the ratio $\kappa$ between the expected cost of the worst case distribution and that of the independent distribution,
i.e.,
\begin{equation}
\label{eq:cordef}
\kappa := \frac{\Ex_{{\cal D}^R}[f(x,S)]}{\Ex_{{\cal D}^I} [f(x,S)]},
\end{equation}
where ${\cal D}^I$ is the independent Bernoulli distribution (also called product distribution) with marginals $\{p_i\}$, and ${\cal D}^R$ is the worst-case distribution (as given by \eqref{primal}). \newline
\\
Suppose that for some particular cost function $f$, the correlation gap can be upper bounded above by $\overline \kappa$ {\it for all $x$}, then it is not difficult to show that the decision obtained assuming independent distribution will give a ${\overline \kappa}$-approximate solution to the corresponding robust optimization problem. More precisely, let $x_I$ is the optimal solution to the stochastic optimization problem \eqref{SP} with independent Bernoulli distribution, and $x_R$ is the optimal solution to the correlation robust problem \eqref{DRSPB}. Then,
$$
\begin{array}{lcl}
g(x_I) & = & \Ex_{{\cal D}^R}[f(x_I,S)], \ \ {\rm and} \\
\noalign{\medskip}
g(x_R) & = & \Ex_{{\cal D}^R}[f(x_R,S)] \ge \Ex_{{\cal D}^I}[f(x_R,S)] \\
& \ge & \Ex_{{\cal D}^I}[f(x_I,S)]
\end{array}
$$
Using the bound on correlation gap at $x_I$, this implies
$$g(x_I) \le \overline \kappa \ g(x_R)$$
Unfortunately, for general cost functions, the correlation gap and hence the corresponding approximation factor can be large in order of $n$, as demonstrated by the following examples.\newline
\\
\example{1}{Minimum cost flow: $\Omega(2^n)$ correlation gap for supermodular functions}
{\\(Sketch)
Consider a two-stage minimum cost flow problem as in Figure \ref{fig:minflow}. There is a single source $s$, and $n$ sinks $t_1,t_2,\ldots, t_n$. Each sink $t_i$ has a probability $p_i=\frac{1}{2}$ to request a demand, and then a unit flow has to be sent from $s$ to $t_i$. Each arc $(u, t_i)$ has a fixed capacity $1$, but the the capacity of arc $(s,u)$ needs to be purchased at a cost $c^I(x)$ in the first stage, and a higher cost $c^{II}(x)$ in the second stage after the set of demand requests is revealed. $c^I(x)$, $c^{II}(x)$ are given as
$$
c^I(x)=\left\{\begin{array}{ll}
x, & x \leq n-1\\
n+2, & x = n \end{array}
\right. ~~~~\ \ \ \ \ \  \ c^{II}(x)=2^n x.
$$
\begin{figure}\label{fig:minflow}
\begin{center}
 \vspace{-0.4in}
 \includegraphics[width=0.38\textwidth,height=0.28\textwidth]{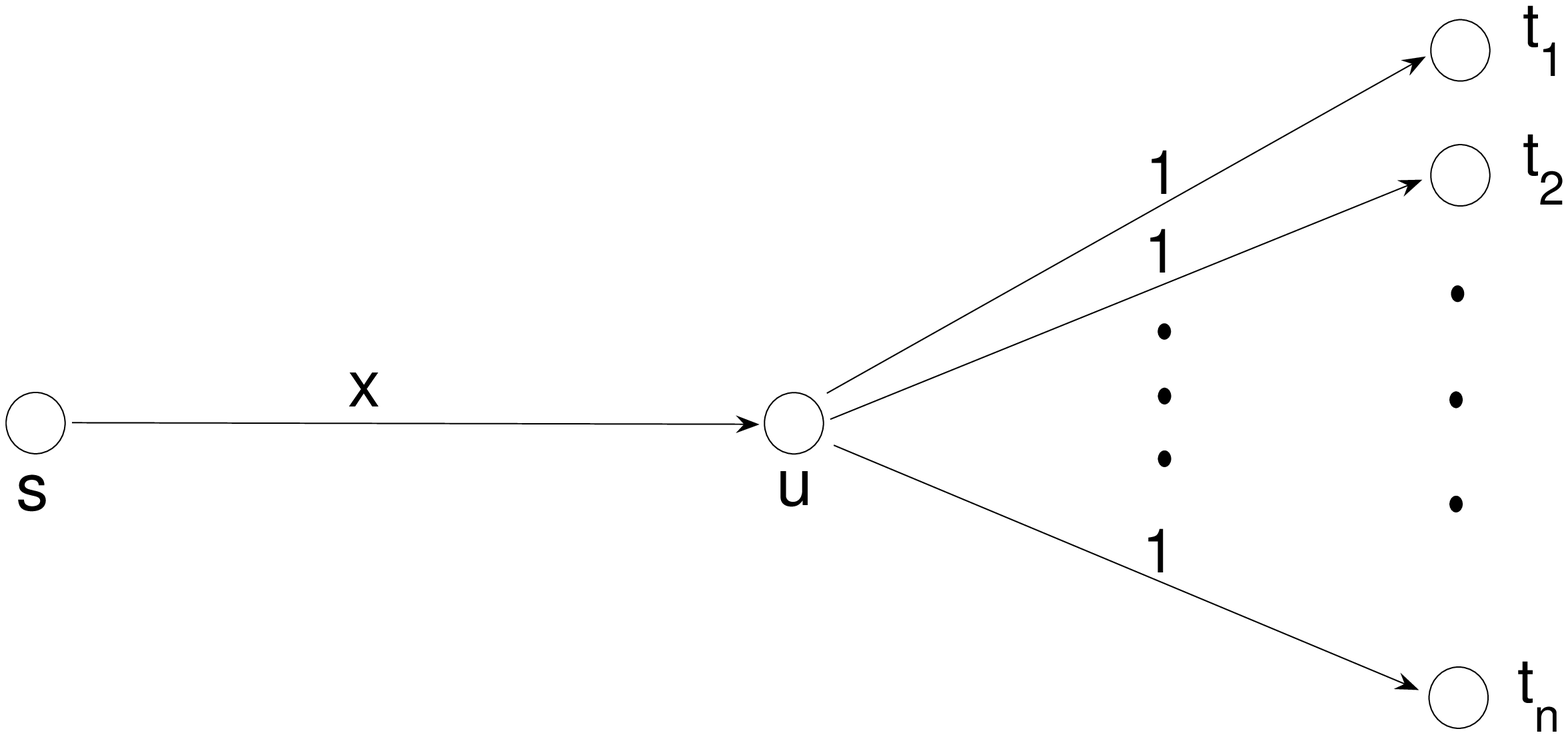}
  \vspace{-0.4in}
  \caption{An example with exponential correlation gap}
\end{center}
\end{figure}

Given the first stage decision $x$, the cost of edges that need to be bought in the second stage to serve a set $S$ of requests is given by: $f(x,S) =c^I(x)+c^{II}(|S|-x)^+ = c^I(x) + 2^n (|S|-x)^+$. It is easy to check that $f(x,S)$ is supermodular in $S$ for any given $x$, i.e. $f(x, S\cup i) - f(x, S) \ge f(x, T \cup i) - f(x, T)$ for any $S \supseteq T$. The objective is to minimize the total expected cost $c^I(x)+\Ex[f(x,S)]$. If the decision maker assumes independent demands from the sinks, then $x_I=n-1$ minimizes the expected cost, and the expected cost is $n$; however, for the worst case distribution the expected cost of this decision will be $g(x_I)=2^{n-1}+n-1$ (when $\Pr(V)=\Pr(\emptyset)=1/2$ and all other scenario have zero probability). 
Hence, the correlation gap at $x_I$ is exponentially high. A risk-averse strategy  is to use the robust solution $x_R=n$, which leads to a cost $g(x_R)=n+1$. Thus, approximation ratio $g(x_I)/g(x_R)=\Omega(2^n)$.} 
\newline
\\
\example{2}{Stochastic set cover: $\Omega(\sqrt{n}\frac{\log\log{n}}{\log{n}})$ correlation gap for subadditive functions}{\\(Sketch)
Consider a set cover problem with elements $V=\{1, \ldots, n\}$. Each item $j \in V$ has a marginal probability of $1/K$ to appear in the random set $S$. The covering sets are defined as follows. Consider a partition of $V$ into $K=\sqrt{n}$ sets $A_1, \ldots, A_K$ each containing $K$ elements. The covering sets are all the sets in the cartesian product $A_1\times \cdots \times A_K$. Each set has unit cost. Then, cost of covering a set $S$ is given by subadditive function
\begin{equation*}
c(S)=\max_{i=1, \ldots, K}  |S\cap A_i| \ \ \ \forall S \subseteq V.
\end{equation*}
The worst case distribution with marginal probabilities $p_i=1/K$ is one where probabilities $\Pr(S)=1/K$ for $S=A_i$, $i=1,2,\ldots,K$, and $\Pr(S)=0$  otherwise. The expected value of $c(S)$ under this distribution is $K=\sqrt{n}$.
For independent distribution, $c(S)=\max_{i=1,\ldots,K} {\zeta_i}$, where $\zeta_i=|S \cap A_i|$ are independent $(K, 1/K)$-binomially distributed random variables. 

As $K$ approaches $\infty$, since expected value of remains fixed at $1$, the Binomial($K$, $1/K$) distribution approaches the Poisson distribution with expected value $1$. Using some known results on maxima of independent poisson random variables in \cite{kimber83}, it can be shown that for large $K$, the expected value of the maximum of $K$ i.i.d. poisson random variables is bounded by $\Theta(\log K/\log\log K)$ (refer to \rmconf{Appendix \ref{app:poisson}} \addconf{\cite{tech-report}} for a detailed proof). This implies that $\Ex [\max_{i=1, \ldots, \sqrt{n}} \{\zeta_i\}]$ is bounded by $\Theta(\log n/\log\log n)$ for large $n$. So the correlation gap is atleast $\Omega({\sqrt{n}}\log\log n/\log{n})$.

To obtain approximation lower bound for two-stage stochastic set cover instance, extend the above instance as follows. For ease of notation, let $L(n) = d\log{n}/\log\log{n}$, where $d$ is a constant such that $\Ex [\max_{i} \{\zeta_i\}] \le L(n)$ . Let the first stage cost of a covering set to be $w^I=(1+\epsilon)L(n)/\sqrt{n}$ for some small $\epsilon>0$, and the second stage cost to be $w^{II}=1$. For a given first stage cover $x$, let $B(x)$ be the set of elements covered by $x$, then $f(x,S)= w^I|x|+c(S-B(x))$. Using above analysis for function $c(S)$, the optimal solution for independent distribution will be to buy no (or very few) sets in the first stage giving $\Ex[f(x,S)] \le L(n)$ for independent distribution, but $\Theta(\sqrt{n})$ cost for worst case distribution. On the other hand, the optimal robust solution considering worst case distribution is to cover all the elements in the first stage giving $\OO(L(n))$ cost in the worst case. Thus, approximation ratio $g(x_I)/g(x_R)=\Omega(\sqrt{n}\log\log{n}/\log{n})$.
}\newline



These examples indicate that using independent distribution may not always give a good approximation to the robust model. However, below we identify a wide class of functions for which correlations may be ignored to get efficient solutions for stochastic optimization problems.

\section{A class of functions with low correlation gap}
A key contribution of our paper is to identify a class of cost functions for which the correlation gap is well bounded. 
To our interest, many popular cost functions including submodular functions, facility location, Steiner forest, etc. belong to this class, which will lead to efficient approximations for these problems. \newline
\\
We derive our characterization using concepts of cost-sharing. A cost-sharing scheme is a function defining how to share the cost of a service among the serviced customers. We consider the class of cost functions $f$ such that for every feasible $x$, there exists some cost-sharing scheme for allocating the cost $f(x,S)$ among members of set $S$ with (a) \emph{$\beta$-budget balance} (b) weak \emph{cross-monotonicity}, and (c) weak \emph{$\eta$-summability}. Below we precisely state these properties.
Since we assume that $x$ can take any fixed value, we will abbreviate $f(x,S)$ as $f(S)$ for simplicity when clear from the context.\newline
\\
 A cost-sharing scheme is {\it cross-monotonic} if it satisfies the property that everyone is better off when the set of people who receive the service expands \cite{agt}. Roughgarden et al \cite{roughgarden} introduced an additional property of {\it summability} for cost-sharing schemes. 
Here, we will define a slightly weaker version of these properties by requiring them to hold only for given ordering on a subset of $V$.  More precisely, we define a cost-sharing scheme as a function $\chi(i, S, \sigma_S)$ that, for each element $i\in S$ and ordering $\sigma_S$ on $S$, specifies the share of $i$ in $S$. 
The three properties of budget-balance, weak cross-monotonicity and weak summability are now stated as follows:
\begin{enumerate}
\item{\it $\beta$-budget balance:} For all $S$, and orderings $\sigma_S$ on $S$:
$$f(S) \ge \sum_{i=1}^{|S|} \chi(i,S, \sigma_S) \ge \frac{f(S)}{\beta}
$$
\item{\it Cross-monotonicity:} For all $i\in S$, $S\subseteq T$, $\sigma_S \subseteq \sigma_T$: 
$$\chi(i,S, \sigma_S) \ge \chi(i,T, \sigma_T)$$
 Here , $\sigma_S \subseteq \sigma_T$ means that the ordering $\sigma_S$ is a restriction of ordering $\sigma_T$ to subset $S$.
\item{\it Weak $\eta$-summability:} For all $S$, and orderings $\sigma_S$:
$$\sum_{\ell=1}^{|S|} \chi(i_{\ell}, S_{\ell}, \sigma_{S_{\ell}}) \le \eta f(S)
$$
where $i_{\ell}$ is the $\ell^{th}$ element and $S_{\ell}$ is the set of the first $\ell$ members of $S$ according to ordering $\sigma_S$. And, $\sigma_{S_{\ell}}$ is the restriction of $\sigma_S$ on $S_{\ell}$.
Note that this is a weaker requirement than the conventional definition of summability, where a single cost-sharing function $\chi(i, S)$ must satisfy the given inequality for {\it all} orderings on the ground set \cite{roughgarden}.
\end{enumerate}
We may re-emphasize that any cost-sharing scheme satisfying the conventional definition of $\beta$-budget-balance, cross-monotonicity and $\eta$-summability (as in \cite{agt, roughgarden}) will always satisfy the above weaker conditions. However, this relaxation to weak conditions could give significant savings in approximation factors for some cases. For example, submodular functions satisfy the above weak conditions with $\eta=1$ and $\beta=1$ for the incremental cost-sharing scheme:
$$ \chi(i,S, \sigma_S) = f(S_{i}) - f(S_{i-1})$$
where $S_{i}$ is the set of the first $i$ members of $S$ according to ordering $\sigma_S$.
On the other hand, for the conventional definition of summability, a lower bound of $\eta\ge \Omega(\log n)$ was shown for submodular functions in \cite{roughgarden}.\newline
\\
Let us call a cost-sharing scheme satisfying the above three properties an $(\eta,\beta)$-cost-sharing scheme. Also, we say that a function $f(x,S)$ is non-decreasing in $S$ if for every $x$ and every $S \subseteq T$, $f(x,S) \le f(x,T)$. 
Our main result is the following theorem, which we will prove in the next section:

\begin{theorem}\label{th:MAIN}
For any instance $(f,V,\{p_i\})$, if for all feasible $x$, the cost function $f(x,S)$ is non-decreasing in $S$ and has an $(\eta, \beta)$-cost-sharing scheme for elements in $S$, then the correlation gap is bounded as $\eta \beta\left(\frac{e}{e-1}\right)$.
\end{theorem}

As described in Section \ref{sec:cor}, this gives following corollary for approximating the correlation robust model:
\begin{corollary}
\label{cor:MAIN}
For instances $(f,V,\{p_i\})$ as defined in Theorem \ref{th:MAIN}, an $\eta\beta \frac{e}{e-1}$ approximate solution for correlation robust optimization problem can be constructed by solving the corresponding stochastic optimization problem under independent distribution.
\end{corollary}
Further, it is easy to show that for these functions, the variance under independent distribution is bounded by $O(\frac{\eta^2\beta^2}{\bar{p}^2})$, where $\bar{p}=\min_i \{p_i\}$. Thus, if the cost function is convex in $x$, these stochastic optimization problems may be solved efficiently using sample average approximation (SAA) method \cite{shapiro02}. For specific problems, the structural simplicity provided by independent distribution may even eliminate the need of using sample average approximation.\newline
\\
Before moving on to the proof of Theorem \ref{th:MAIN}, let us briefly discuss its  implications for various stochastic optimization problems, and 
for a seemingly unrelated problem of welfare maximization in combinatorial auctions:


\subsection {Stochastic optimization with submodular functions}
A function $h:2^V\rightarrow \R$ is submodular if $h(S\cup i) - h(S) \le h(T \cup i) - h(T)$ for all $S\supseteq T$, and $i\in V$. These cost functions are characterized by diminishing marginal costs, which is common for resource allocation problems where a resource can be shared by multiple users and thereby the marginal cost decreases as number of users increases. As discussed earlier, for submodular functions $\eta=1,\beta=1$.
Therefore, Theorem \ref{th:MAIN} directly leads to the following corollary:

\begin{corollary}
\label{th:subm-approx}
If the cost function $f(x,S)$ is non-decreasing and submodular in $S$ for all feasible $x$, then for any instance $(f,V,\{p_i\})$,
the correlation gap is bounded by the constant $\frac{e}{e-1}$.
\end{corollary}

The next example shows the $e/(e-1)$ bound is tight for submodular functions.\newline
\\
\example{3}{Tightness}{
Let $V:=\{1,2,\ldots,n\}$, define $f(S)=1$ if $S \ne \emptyset$, and $f(\emptyset)=0$. Let each item has a probability $p=\frac{1}{n}$. Then the worst case distribution is $Pr({\{i\}})=1/n$ for each $i \in V$, with expected value $1$. The independent distribution has an expected cost $1-(1-\frac{1}{n})^n \rightarrow 1-1/e$ as $n \rightarrow \infty$.
}




\subsection{Stochastic Uncapacitated Facility Location (SUFL)} In two-stage stochastic facility location problem, any facility $j\in F$ can be bought at a low cost $w^I_j$ in the first stage, and higher cost $w^{II}_j > w^I_j$ in the second stage, that is, after the random set $S \subseteq V$ of cities to be served is revealed. The decision maker's problem is to decide $x \in \{0,1\}^{|F|}$, the facilities to be build in the first stage so that the total expected cost
$\Ex[f(x,S)]$ of facility location is minimized (refer to \cite{swamy} for further details on the problem definition).

Given a first stage decision $x$, the cost function $f(x,S)=w^I\cdot x+c(x,S)$, where $c(x,S)$ is the cost of deterministic UFL for set $S \subseteq V$ of customers and set $F$ of facilities such that the facilities $x$ already bought in first stage are available freely at no cost, while any other facility $j$ costs $w^{II}_j$. 
For this deterministic UFL cost function there exists a cross-monotonic, $3$-budget balanced, $\log |S|$ summable cost-sharing scheme \cite{roughgarden07}. Therefore, using Theorem 1, we get following bound on correlation gap:

\begin{corollary}
The correlation gap for Stochastic uncapacitated facility location is bounded by $O(\log n)$, where $n=|V|$, the number of cities to be served.
\end{corollary}

This observation reduces our robust facility location problem to the well-studied stochastic UFL problem under known (independent Bernoulli) distribution \cite{swamy} at the expense of an $O(\log n)$ approximation factor.

\subsection{Stochastic Steiner Tree (SST)}
In the two-stage stochastic Steiner tree problem, we are given a graph $G = (V,E)$. An edge $e\in E$ can be bought at cost $w^I_e$ in the first stage. The random set $S$ of terminals to be connected are revealed in the second stage. More edges may be bought at a higher cost $w^{II}_e, e\in E$ in the second stage after observing the actual set of terminals. Here,  decision variable $x$ is the edges to be bought in the first stage, and cost function $f(x,S)=w^I\cdot x + c(x,S)$, where $c(x,S)$ is the Steiner tree cost function for set $S$ given that the edges in $x$ are already bought.
Since a $\log^2(|S|)$-summable, $2$-budget balanced cost sharing method is known for this cost function \cite{roughgarden07, timKLS}, we can conclude:
\begin{corollary}
The correlation gap for Stochastic Steiner tree is bounded by $O(\log^2 n)$, where $n=|V|$, the number of terminals to be connected.
\end{corollary}
This observation reduces our robust problem to the well-studied (for example see \cite{boosted}) SST problem under known (independent Bernoulli) distribution at the expense of an $O(\log^2 n)$-approximation factor.

\subsection{Welfare Maximization Problem}
Finally, Theorem \ref{th:MAIN} extends some existing results for {\it social welfare maximization} in combinatorial auctions.
Consider the problem of maximizing total utility achieved by partitioning $n$ goods among $K$ players each with utility function $f(S)$ for subset $S$ of goods \footnote{A more general formulation of this problem that is often considered in the literature allows non-identical utility functions for various players.}. The optimal welfare $\OPT$ is obtained by following integer program:
\begin{equation}
\label{eq:LP2}
\begin{array}{llll}
\max_{\alpha} & \sum_{S}{\alpha_S f(S)}\\
\mbox{s.t.} & \sum_{S: i \in S}{\alpha_S} = 1, & \forall i \in V \\
& \sum_{S} {\alpha_S} = K \\
& \alpha_S \in \{0,1\},  & \forall S \subseteq V 
\end{array}
\end{equation}
Observe that on relaxing the integrality constraints on $\alpha$ and scaling it by $1/K$, the above problem reduces to that of finding the worst-case distribution $\alpha^*$ (i.e. one that maximizes expected value $\sum_{S}\alpha_S f(S)$ of function $f$) such that the marginal probability $\sum_{S:i\in S} \alpha_S$ of each element is $1/K$. Therefore:
$$ \OPT \le \Ex_{\alpha^*}[K f(S)] $$
Consequently, the correlation gap bound in Theorem \ref{th:MAIN} leads to the following corollary for welfare maximization problems:

\begin{corollary}
For welfare maximization problems with $n$ goods and $K$ players with identical utility functions $f$, the randomized algorithm that assigns goods independently to each of the $K$ players with probability $1/K$ gives $\frac{1}{\eta\beta} (1-\frac{1}{e})$ approximation to the optimal partition; given that function $f$ is non-decreasing and admits an $(\eta, \beta)$-cost-sharing scheme.
\end{corollary}

Since $\eta=1, \beta=1$ for submodular functions, the above result matches the $1-1/e$ approximation factor provided by Vondrak \cite{vondrak08} for this problem in case of identical monotone submodular functions.

The reader may observe that even though approximating the worst case distribution directly provides a matching approximation for the corresponding welfare maximization problem, the converse is not true. In addition to having uniform probabilities $p_i=1/K$, solutions for welfare maximization approximate the integer program \eqref{eq:LP2}, where as the worst case distribution requires solving the corresponding LP relaxation. The latter is a strictly harder problem unless the integrality gap is $0$. A notable example is the above-mentioned case of identical submodular functions. This case was studied by Bhikchandani \cite{bikhchandani97} in context of Walrasian equilibria who conjectured a $0$ integrality gap for this problem implying the existence of Walrasian equilibria. However, in \rmconf{appendix \ref{app:integrality}} \addconf{\cite{tech-report}}, we show a simple counter-example with non-zero integrality gap ($11/12$) for this problem. As a byproduct, this counter-example proves that even for identical submodular valuation functions, Walrasian equilibria may not exist.
\section{Proof of Theorem \ref{th:MAIN}}
\label{sec:proof}
For a problem instance $(f,V,\{p_i\})$ and fixed $x$, use $\LP(f,V,\{p_i\})$ and $\INDEP (f,V,\{p_i\})$ to denote the expected cost of worst-case distribution and independent Bernoulli distribution respectively. In this section, we prove our main technical result that the correlation gap $$\frac{\LP(f,V,\{p_i\})}{\INDEP (f,V,\{p_i\})} \le \eta\beta \frac{e}{(e-1)}$$ 
when $f$ is non-decreasing and admits $(\eta, \beta)$ cost-sharing in $S$. As before, we will abbreviate $f(x,S)$ as $f(S)$ for simplicity. 

The proof is structured as follows. We first focus on special instances of the problem in which all $p_i$'s are equal to $1/K$ for some integer $K$, and the worst case distribution is a ``K-partition-type" distribution. That is, the worst case distribution divides the elements of $V$ into $K$ disjoint sets $\{A_1, \ldots, A_K\}$, and each $A_k$ occurs with probability $1/K$. Observe that for such instances, the expected value under worst case distribution is  $\LP(f,V,\{p_i\})=\frac{1}{K} \sum_k f(A_k)$. In Lemma \ref{th:alphamain}, we show that for such ``nice" instances the correlation gap is bounded by $\eta \beta\frac{e}{e-1}$. Then, we use a ``split" operation to reduce any given instance of our problem to a nice instance such that the reduction can only increase the correlation gap. This will show that the bound $\eta \beta\frac{e}{e-1}$ for nice instances is an upper bound for any instance of the problem, thus concluding the proof of the theorem.

\begin{lemma}\label{th:alphamain}
For instances $(f,V,\{p_i\})$ such that
(a) $f(S)$ is non-decreasing and admits an $(\eta, \beta)$-cost-sharing scheme (b) marginal probabilities $p_i$ are all equal to $1/K$ for some integer $K$, and (c) the worst case distribution is a $K$-partition-type distribution, the correlation gap is bounded as:
\[
\frac{\LP(f,V,\{1/K\})}{\INDEP (f,V,\{1/K\})} \leq \eta\beta \frac{e}{(e-1)}
\]
\end{lemma}
\begin{proof}
Let the optimal $K$-partition corresponding to the worst case distribution is $\{A_1,~A_2,~\ldots,A_K\}$. Assume w.l.o.g that $f(A_1) \ge f(A_2) \ge \ldots \ge f(A_K)$. 
Fix an order $\sigma$ on elements of $V$ such that for all $k$, the elements in $A_k$ come before $A_{k-1}$. For every set $S$, let $\sigma_S$ be the restriction of ordering $\sigma$ on set elements of set $S$.
 Let $\chi$ is the $(\eta, \beta)$ cost-sharing scheme for function $f$, as per the assumptions of the lemma. 
Then by weak $\eta$-summability of $\chi$:

\begin{equation}\label{eq:eta1}
\begin{array}{lcl}
\INDEP(f,V,\{1/K\}) & =& \E_{S\subseteq V}[f(S)] \\
\noalign{\medskip}
		    & \geq & \frac{1}{\eta} \ \E_{S\subseteq V} \big[\sum_{l=1}^{|S|} \chi(i_l,S_l, \sigma_{S_l})\big]
\end{array}
\end{equation}
where the expected value is taken over independent distribution.

Denote $\phi(V):= \E_{S\subseteq V} \big[\sum_{l=1}^{|S|} \chi(i_l,S_l,\sigma_{S_l})\big]$. Let $p=1/K$. We will show that
$$ \phi(V) \ge (1-p)\phi(V\backslash A_1) + \frac{1}{\beta} f(A_1)$$
Recursively using this inequality will prove the result.
To prove this inequality, denote $S_{-1}=S \cap (V\backslash A_1)$, $S_1=S \cap A_1$, for any $S \subseteq V$. Since elements in $A_1$ come after the elements in $V\backslash A_1$ in ordering $\sigma_S$, note that for any $\ell\le |S_{-1}|$, $S_{\ell} \subseteq S_{-1}$, and for $\ell >|S_{-1}|$, $i_{\ell} \in S_1$.
\begin{equation}\label{eq:eta2}
\begin{array}{rcl}
\phi(V) & = & \E_{S}\big[\sum_{l=1}^{|S_{-1}|}\chi(i_l,S_l, \sigma_{S_l})\big]\\
\noalign{\medskip}
&  & + \ \E_{S}\big[\sum_{l=|S_{-1}|+1}^{|S|}\chi(i_l,S_l, \sigma_{S_l})\big]
\end{array}
\end{equation}
Since $S_{\ell} \subseteq S \cup A_1$, using cross-monotonicity of $\chi$, the second term above can be bounded as:
\begin{equation}
\label{eq:eta0}
\begin{array}{l}
\E_{S} [\sum_{l=|S_{-1}|+1}^{|S|}\chi(i_l,S_l, \sigma_{S_l})]  \\
\noalign{\medskip}
\ \ge \  \E_{S } [\sum_{l=|S_{-1}|+1}^{|S|}\chi(i_l,S \cup A_1 , \sigma_{S \cup A_1})]
\end{array}
\end{equation}
Because $S_{-1}$ and $S_{1}$ are mutually independent, for any fixed $S_{-1}$, each $i \in A_1$ will have the same conditional probability $p=1/K$ of appearing in $S_1$. Therefore,

\begin{equation}\label{eq:eta3}
\begin{array}{l}
\E_{S} \big[\sum_{l=|S_{-1}|+1}^{|S|}\chi(i_l,S \cup A_1, \sigma_{S  \cup A_1})\big] \\ 
\noalign{\medskip}
= \  \E_{S_{-1}}\big[\E_{S_1}[\sum_{l=|S_{-1}|+1}^{|S|}\chi(i_l,S_{-1} \cup A_1, \sigma_{S_{-1} \cup A_1})|S_{-1}]\big]\\
\noalign{\medskip}
= \  p \ \E_{S_{-1}}[ \sum_{i \in A_1} {\chi(i,S_{-1} \cup A_1, \sigma_{S_{-1} \cup A_1})}]\\
\end{array}
\end{equation}

Again, using independence and cross-monotonicity, analyze the first term in the right hand side of \eqref{eq:eta2},

\begin{equation} \label{eq:eta4}
\begin{array}{l}
\E_{S} [ \sum_{l=1}^{|S_{-1}|}\chi(i_l,S_l, \sigma_{S_l})] \\
\noalign{\medskip}
= \ \E_{S_{-1}}[\sum_{l=1}^{|S_{-1}|} \chi(i_l,S_l, \sigma_{S_l})] \\
\noalign{\medskip}
\geq \ (1-p)\ \E_{S_{-1}}[\sum_{l=1}^{|S_{-1}|} \chi(i_l,S_l, \sigma_{S_l})]\\
\noalign{\medskip}
\ \ \ + \ p \ \E_{S_{-1}}[\sum_{l=1}^{|S_{-1}|} \chi(i_l,S_{-1} \cup A_1, \sigma_{S_{-1} \cup A_1})] \\
\noalign{\medskip}
= \ (1-p) \ \phi(V \backslash A_1) \\
\noalign{\medskip}
\ \ \ + \ p \ \E_{S_{-1} }[\sum_{l=1}^{|S_{-1}|} \chi(i_l,S_{-1} \cup A_1, \sigma_{S_{-1} \cup A_1})]
\end{array}
\end{equation}
Based on \eqref{eq:eta2}, \eqref{eq:eta3} and \eqref{eq:eta4}, and the fact that the cost-sharing scheme $\chi$ is $\beta$-budget balanced, we deduce

\begin{equation}
\begin{array}{rcl}
\phi(V)&=&(1-p) \ \phi(V \backslash A_1) \\ 
\noalign{\medskip}
& & \ + \ p \ \E_{S_{-1}}[\sum_{l=1}^{|S_{-1}|} \chi(i_l,S_{-1} \cup A_1, \sigma_{S_{-1} \cup A_1}) +\\
\noalign{\medskip}
& & \ \ \ \ \ \ \ \ \ \ \ \ \ \ \sum_{i \in A_1}{\chi(i,S_{-1}\cup A_1, \sigma_{S_{-1}\cup A_1})}] \\
\noalign{\medskip}
& \ge & (1-p)\ \phi(V \backslash A_1) + \frac{1}{\beta} p \ \E_{S_{-1}} [f(S_{-1} \cup A_1)] \\
\noalign{\medskip}
&\geq & (1-p)\ \phi(V \backslash A_1)+ \frac{1}{\beta} p \ f(A_1), 
\end{array}
\end{equation}
The last inequality follows from monotonicity of $f$. Expanding the above recursive inequality for $A_2$, $\ldots$, $A_K$, we get
\beq\label{eq:eta5}
\phi(V) \geq \frac{1}{\beta} p \sum_{k=1}^K { (1-p)^{k-1} f(A_k)},
\eeq
Since $f(A_k)$ is decreasing in $k$, and $p=1/K$ by simple arithmetic one can show
$$
\begin{array}{rcl} 
\phi(V) & \ge & \frac{1}{\beta} \cdot \sum_{k=1}^K pf(A_k) \cdot \frac{(\sum_{k=1}^K (1-p)^{k-1})}{K} \\
\noalign{\medskip}
& \ge & \frac{1}{\beta} \cdot (1-\frac{1}{e}) \cdot \sum_{k=1}^K pf(A_k)
\end{array}
$$
By definition of $\phi(V)$, this gives:
$$
\INDEP (f,V,\{1/K\}) \geq \frac{1}{\eta \beta} \left(1-\frac{1}{e}\right) \LP(f,V,\{1/K\}).
$$

\end{proof}
Next, we reduce a general problem instance to an instance satisfying the properties required in Lemma \ref{th:alphamain}.
We use the following split operation.\newline
\paragraph{Split:}
Given a problem instance $(f,V,\{p_i\})$, and integers $\{n_i \ge 1, i\in V\}$, define a new instance $(f', V', \{p'_j\})$ as follows: split each item $i\in V$ into $n_i$ copies $C^i_1,C^i_2,\ldots, C^i_{n_i}$, and assign a marginal probability of $p'_{C^i_k} = \frac{p_i}{n_i}$ to each copy.
Let $V'$ denote the new ground set containing all the duplicates. Define the new cost function  $f':2^{V'}\rightarrow \R$ as:
\beq\label{eq:f'}
 f'(S') = f(\Pi(S')), \mbox{ for all $S' \subseteq V'$ },
\eeq
where $\Pi(S') \subseteq V$ is the original subset of elements whose duplicates appear in $S'$,
i.e. $\Pi(S') = \{i \in V| C^i_k \in S'~ \mbox{for some}~ k \in \{1,2,\ldots,n_i\}\}$.

The split operation has following properties. Their proofs \rmconf{will be given in Appendix \ref{app:split}} \addconf{appear in \cite{tech-report}}.

\begin{property} \label{prop1}
If $f(S)$ is a non-decreasing function in $S$, 
then so is $f'$.
\end{property}
\begin{property}\label{prop2}
If $f(S)$ is non-decreasing in $S$,
then splitting does not change the worst case expected value, that is:
$$
\LP(f,V, \{p_i\})=\LP(f',V',\{p'_j\})
$$
\end{property}
\begin{property}\label{prop3}
If $f(S)$ is non-decreasing in $S$, then splitting can only decrease the expected value over independent distribution:
$$
\INDEP(f,V,\{p_i\}) \ge \INDEP(f',V',\{p'_j\}) .
$$
\end{property}

The remaining proof tries to use these properties of split operation for reducing any given instance to a ``nice" instance so that Lemma \ref{th:alphamain} can be invoked for proving the correlation gap bound. \newline\\
\noindent
{\it Proof of Theorem \ref{th:MAIN}.}
Suppose that the worst case distribution for instance $(f,V,\{p_i\})$  is not a partition-type distribution. Then, split any element $i$ that appears in two different sets. Simultaneously, split the distribution by assigning probability $\alpha_{S'}=\alpha_{\Pi(S')}$ to the each set $S'$ that contains exactly one copy of $i$. Repeat  until the distribution becomes a partition. 
Since each new set in the new distribution contains exactly one copy of $i$, 
by definition of function $f'$, this splitting does not change the expected function value.  
By Property \ref{prop2} of Split operation, the worst case expected values for the two instances (before and after splitting) must be the same, so this partition forms a worst case distribution for the new instance. 
 Then, we further split each element (and simultaneously the distribution) until such that the marginal probability of each new element is $1/K$ for some large enough integer $K$ \footnote{Such an integer $K$ can always be reached assuming $p_i$s are rational.}. 
This reduces the worst case distribution to a partition $A_1, \ldots, A_K$ such that each set $A_k$ has probability $1/K$. Thus, the conditions (b) and (c) of Lemma \ref{th:alphamain} are satisfied by the reduced instance $(f',V', \{p'_i\})$. 

By the properties \ref{prop2}, \ref{prop3} of Split operation, the correlation gap can only becomes larger on splitting. So, we can focus on proving the correlation gap bound for the new instance.
Now, let us consider the remaining condition (a) of Lemma \ref{th:alphamain}. By Property \ref{prop1}, the cost function $f'$ obtained by splitting is non-decreasing. 
Given the original $(\eta,\beta)$ cost-sharing method $\chi$ for $f$, we show that there exists a  cost-sharing method $\chi'$ for the new instance such that $\chi'$ is (1) $\beta$-budget balanced (2) weak $\eta$-summable, and (3) cross monotone in following weaker sense. 
$\chi'$ is cross-monotone for any $S' \subseteq T', \sigma_{S'} \subseteq \sigma_{T'}$ such that $\sigma_{S'},\sigma_{T'}$ respect the partial order $A_{K}, \ldots, A_1$ of elements, and $S'$ is a {\it partial-prefix} of $T'$, that is, for some $k\in \{1, \ldots, K\}$, $S' \subseteq A_K \cup \cdots \cup A_{k}$, and $T'\backslash S' \subseteq  A_{k} \cup \cdots \cup A_1$. The construction of this cost-sharing scheme is given in appendix, Lemma \ref{lem:new-cost-share}.

Thus, all the conditions in Lemma \ref{th:alphamain} are satisfied by the new instance except for the cross-monotonicity. The weaker cross-monotonicity that the new instance satisfies is actually sufficient to prove Lemma \ref{th:alphamain}. To see this, observe that cross monotonicity is used only in Equation \ref{eq:eta0} and \ref{eq:eta4}, and at both of these places, the required prefix condition is satisfied. Thus, Lemma \ref{th:alphamain} can be invoked to bound the correlation gap for the new instance, thereby completing the proof. \qed

\section{Supermodular functions}
In the end, we directly consider the correlation robust model for cost functions $f(x,S)$ which are supermodular in $S$. As shown in Section \ref{sec:cor}, the correlation gap for these cost functions can be exponentially high, so independent distribution does not give a good approximation to the worst case distribution. However, it is easy to characterize the worst case distribution and directly solve the correlation robust model in this case.

\begin{lemma}
\label{th:supermodular}
Given that function $f: 2^V \rightarrow \R$ is supermodular, the worst case distribution over $S$ has the following closed form
$$
\Pr(S) = \left\{ \begin{array}{ll}
	p_n & \mbox{ if } S=S_n\\
	p_i - p_{i+1} & \mbox{ if } S=S_i, 1 \le i \le n-1\\
	1-p_1				& \mbox{ if } S = \emptyset\\
	0 					& o.w.
	\end{array}
	\right.
$$

where $n=|V|$; $i$ is the $i^{th}$ member of $V$ and $S_i$ is the set of first $i$ members of $V$, both with respect to a specific ordering over $V$ such that $p_1\ge \ldots \ge p_n$.
\end{lemma}

The lemma is simple to prove, a proof appears in \rmconf{appendix \ref{app:supermodular}} \addconf{\cite{tech-report}}.
Lemma \ref{th:supermodular} implies following corollary for solving the robust optimization problem.
\begin{corollarylemma}
\label{cor:supermodular}
For cost functions $f(x,S)$ that are supermodular in $S$ for any feasible $x$, the robust optimization problem is simply formulated as:
$$ \min_{x\in C} p_n f(x,S^n) + \sum_{i=1}^{n-1} (p_i - p_{i+1}) f(x,S^{i}) + (1 - p_1) f(x,\phi)
$$
\end{corollarylemma}
Thus, if $f(x,S)$ is convex in $x$ and $C$ is a convex set, then it is a convex optimization problem and can be solved efficiently. \newline
\\
\paragraph{Acknowledgements} The authors would like to thank Ashish Goel and Mukund Sundarajan for many useful insights on the problem.

\bibliographystyle{abbrv}
\bibliography{soda}

\appendix

\rmconf{
\section{Maximum of Poisson Random Variables}
\label{app:poisson}
In this section, we show that the expected value of the maximum of a set of $M$ independent identically distributed poisson random variables can be bounded as $O(\log M/\log \log M)$ for large $M$.

Let $\lambda$ denote the mean, and $F$ denote the distribution of i.i.d. poisson variables $X_i$. Define $G=1-F$. Also define continuous extension of $G$:
$$G_c(x)=\exp(-\lambda) \sum_{j=1}^{\infty} \lambda^{(x+j)}/\Gamma(x+j+1)$$
Note that $G(k)=G_c(k)$ for any non-negative integer $k$. 
Let $\{A_k\}_{k=1}^{\infty}$ is defined by $G_c(A_k) = 1/k$. Define continuous function $L(x) = \log(x)/\log\log(x)$. Then, in \cite{kimber83}, it is shown that for large $k$, $A_k \sim L(k)$. 

We use these asymptotic results to derive a bound on expectation of $Z=\max_{i=1, \ldots, M} X_i$ for large $M$. 
\begin{eqnarray} 
\label{expected-max}
\Ex[Z] & = & \sum_{k=0}^{\infty} \Pr(Z>k) \nonumber \\
	& = & \sum_{k=0}^{\lceil L(M^2) \rceil} \Pr(Z>k) + \sum_{k=\lceil L({M^2}) \rceil+1}^{\infty} \Pr(Z>k) \nonumber\\
& \le & L({M^2})+1 + \displaystyle\int_{x=L({M^2})}^{\infty} \Pr(Z>x)dx
\end{eqnarray}
Next, we show that the integral term on the right hand side is bounded by a constant for large $M$. Substituting $x=L(y)$ in the integration on the right hand side, we get
\begin{eqnarray*} 
& & \displaystyle\int_{x=L({M^2})}^{\infty} \Pr(Z>x)dx \\
& = & \displaystyle\int_{L(y)=L(M^2)}^{\infty} \Pr(Z>L(y)) L'(y) dy\\
& \le &  \displaystyle\int_{y=M^2}^{\infty} \Pr(Z>L(y))\frac{1}{y} dy
\end{eqnarray*}
$L'(y)$ denotes the derivative of function $L(y)$. The last step follows because $L'(y) \le \frac{1}{y}$ for large enough $y$ (i.e. if $\log\log y\ge 1$).  
Further, since $\frac{\Pr(Z>L(k))}{k}$ is a decreasing function in $k$, it follows that:
\begin{eqnarray*} 
\displaystyle\int_{y=M^2}^{\infty} \frac{\Pr(Z>L(y))}{y} dy & \le & \sum_{k=M^2}^{\infty} \frac{\Pr(Z>L(k))}{k}\\
\end{eqnarray*}
Now, for large $k$, $L(k) \sim A_k$, and
$$ \Pr(Z>A_k) \le 1-(1-G_c(A_k))^M = 1-\left(1-\frac{1}{k}\right)^M$$
Therefore, for large $M$, 
\begin{eqnarray*} 
\sum_{k=M^2}^{\infty} \frac{\Pr(Z>L(k))}{k} & \le & \sum_{k=M^2}^{\infty} \frac{1}{k}-\frac{1}{k}\left(1-\frac{1}{k}\right)^M \\
		& \le &   \sum_{k=M^2}^{\infty} \frac{2M}{k^2}\\
				& \le & 1
\end{eqnarray*}
This proves that the integral term on the right hand side of \eqref{expected-max} is bounded by a constant, and thus, for large $M$:
$$\Ex[Z] \le L(M^2)+ 2 = O(\log M/\log \log M)$$.
\section{Properties of Split Operation}
\label{app:split}
\paragraph{Property \ref{prop1}}
If $f(S)$ is non-decreasing in $S$ with an $(\eta, \beta)$-cost sharing scheme, then so is $f'$.
\begin{proof}
Monotonicity holds since for any $S'\subseteq T'\subseteq V'$, $\Pi(S') \subseteq \Pi(T')$:
$$f'(S') = f(\Pi(S')) \le f(\Pi(T')) = f'(T')$$
\end{proof}
\paragraph{Property \ref{prop2}}
If the cost function $f(\cdot)$ is non-decreasing in $S$,
then the splitting procedure does not change the worst-case expected value. That is:
$$
\LP(f,V, \{p_i\})=\LP(f',V',\{p'_j\})
$$

\begin{proof}
For any fixed $x$, the worst case expected cost is the optimal value of following linear program, where $\{\alpha_S\}_{S\subseteq V}$ represents a distribution over subsets of set $V$:
\begin{equation}\label{eq:LPprop2}
\begin{array}{lrl}
\LP(f,V,\{p_i\}) =& \max_{\alpha} & \sum_{S}{\alpha_S f(x,S)}\\
& \mbox{s.t.} & \sum_{S:~i \in S}{\alpha_S} = p_i,~\forall i \in V \\
& & \sum_{S} {\alpha_S}=1 \\
& & \alpha_S \geq 0,~ \forall S \subseteq V. \\
\end{array}
\end{equation}
Suppose item $1$ is split into $n_1$ pieces, and each piece is assigned a probability $\frac{p_1}{n_1}$. Let $\{\alpha_S\}$ denote the optimal solution for the instance $(f,V,\{p_i\})$, then we can construct a solution for the new instance $(f',V', \{p'_j\})$ which has the same objective value by assigning non-zero probabilities to only those sets which have no duplicates.
$$
\begin{array}{l}
\forall S' \subseteq V', ~~\\
\alpha'_{S'} = \left\{ \begin{array}{rl}
\alpha_{S'}, & \mbox{if $S'$ contains no copies of item $1$} \\
\frac{p_1}{n_1} \alpha_{S'}, & \mbox{if $S'$ contains exactly one copy} \\
0,            & \mbox{otherwise}
\end{array} \right.
\end{array}
$$
One can verify that $\{\alpha'_{S'}\}$ is a feasible distribution (i.e., feasible to the linear program \eqref{eq:LPprop2}) for the new instance $(f',V', \{p'_j\})$, and has the same objective value as $\LP(f,V,\{p_i\})$. Hence, $\LP(f,V, \{p_i\}) \leq \LP(f',V',\{p'_j\})$.

For the other direction, consider an optimal solution $\{\alpha'_{S'}\}$ of the new instance. It is easy to see that there exists an optimal solution $\{\alpha'_{S'}\}$ that $\alpha'_{S'}=0$ for all $S'$ that contain more than one copy of item $1$. To see this, assume for contradiction that some set with non-zero probability has two copies of item 1. By definition of $f'$, removing one copy will not decrease the function value. 
Then, because of monotonicity of $f'$, we can move out one copy to another set $T$ that has no copy of  item $1$. Such $T$ always exists since the probabilities of copies of item $1$ must sum up to $p_1 \le 1$. So, we can  assume that in the optimal solution $\alpha'_{S'}=0$ for any set $S'$ containing more than one copy. Thus, we can set $\alpha_S = \alpha'_{S'}$ where $S$ is the corresponding original set for any $S \subseteq V$. That forms a feasible solution for original instance with same objective value as $\LP(f',V',\{p'_j\})$.  We can apply the argument recursively for all the items to prove the lemma.
\end{proof}

Next, we prove that the expected cost under independent Bernoulli distribution can only decrease by the split operation.

\paragraph{Property \ref{prop3}}
If $f(\cdot)$ is non-decreasing, then after splitting
$$
\INDEP(f',V',\{p'_j\}) \leq \INDEP(f,V,\{p_i\}).
$$

\begin{proof}
Let $(f',V',\{p'_j\})$ denote the new instance by splitting item $1$ into $n_1$ pieces. Denote
\[
\Lambda :=\{ S' \subseteq V'| S'~\mbox{contains at least one copy of }1\},
\]
and denote $\pi=\Pr(S' \in \Lambda)$. Consider the expected cost under independent Bernoulli distribution, by independence,
$$
\begin{array}{rcl}
& & \INDEP(f',V',\{p'_j\})\\
&=& \Ex_{S'} \left[f'(S')\ I(S' \in \Lambda)\right] + \Ex_{S'}\left[f'(S')\ I(S' \notin \Lambda)\right]\\
               &=&  \pi \ \Ex_{S \subseteq V \backslash \{1\} } [f(S \cup \{1\})]\\
               & & \ \ + \ (1-\pi )\ \Ex_{S \subseteq V \backslash \{1\} }[ f(S)] \\
               &\leq & p_1 \ \Ex_{S \subseteq V \backslash \{1\} } [f(S \cup \{1\})]\\
               & & \ \ +\ (1-p_1 )\ \Ex_{S \subseteq V \backslash \{1\} }[ f(S)]\\
               &=& \INDEP(f,V,\{p_i\}).
               \end{array}
$$
The second last inequality holds because $\pi = 1-(1-\frac{p_1}{n_1})^{n_1} \leq p_1$, and $f(S)\leq f(S\cup \{1\})$ by monotonicity.
\end{proof}

\section{$\frac{11}{12}$ Integrality gap for SWM with identical submodular valuations}
\label{app:integrality}
Let $V =\{1,2,3,4,5,6\}$, $K=3$, and construct a monotone submodular value function as
$$
f(S)=\left\{ \begin{array}{ll}
0 & \mbox{if}~S=\emptyset\\
2 & \mbox{if}~|S|=1\\
3 & \mbox{if}~|S \cap \{1,2,3\}|=1 ~\mbox{and}~ |S \cap \{4,5,6\} |=1 \\
4 & \mbox{if}~|S \cap \{1,2,3\}|\ge 2~ \mbox{or}~ |S \cap\{4,5,6\}| \ge 2 \\
\end{array}
\right.
$$
Then the optimal fractional solution to the LP relaxation of \eqref{eq:LP2} is given by $$\alpha_{\{1,2\}}=\alpha_{\{2,3\}}=\alpha_{\{1,3\}}=0.5,~~\alpha_{\{4,5\}}=\alpha_{\{5,6\}}=\alpha_{\{4,6\}}=0.5,$$
with an optimal value $12$; but the optimal integer solution will have an optimal value $11$. So there is an $11/12$ integrality gap.
}

\section{Construction of cost-sharing scheme}

\begin{lemma}
\label{lem:new-cost-share}
Given $(\eta,\beta)$ cost-sharing scheme $\chi$ for $(f,V,\{p_i\})$, there exists a cost-sharing scheme $\chi'$ for instance $(f',V',\{p'_i\})$ constructed by splitting in Section \ref{sec:proof}, such that $\chi'$ is 
(a) $\beta$-budget balanced (b) weak $\eta$-summable, and (c) cross monotone for any $S' \subseteq T'$, $\sigma_{S'} \subseteq \sigma_{T'}$ such that $S'$ is a {\it partial prefix} of $T'$.
\end{lemma}
\begin{proof}
Given cost-sharing scheme $\chi$, construct $\chi'$ as follows: Cost-share $\chi'$ coincides with the original scheme $\chi$ for the sets without duplicates, but for a set with duplicates, it assigns the cost-share solely to the copy with smallest index (as per the input ordering). 
That is, any $S' \subseteq V'$, ordering $\sigma'_{S'}$, and item $C^i_j$ ($j$-th copy of item $i$) in $S'$, allocate cost-shares as follows:
\begin{equation}
\chi'(C^i_j, S', \sigma'_{S'})=\left\{\begin{array}{ll}
\chi(i, S, \sigma_S), & j=\min\{h:~C^i_h \in S'\}, \\
0, & \mbox{o.w.}
\end{array}\right.
\end{equation}
where $S=\Pi(S')$, $\sigma_S$ is the ordering of lowest index copies in $\sigma'_{S'}$, and $min$ is taken with respect to the ordering $\sigma'_{S'}$.
It is easy to see that the property of {\it $\beta$-budget-balance} carries through to the new cost sharing scheme. 
{\it Weak $\eta$-summability} holds since
 $$
 \begin{array}{ll}
 \displaystyle\sum_{\ell=1}^{|S'|}\chi'(i'_{\ell},S'_{\ell}, \sigma_{{S'}_{\ell}}) = \displaystyle\sum_{j=1}^{|S|}\chi(i_j, S_j, \sigma_{S_j}) & \le \eta f(S) \\
 & = \eta f'(S')
 \end{array}$$
where $S=\Pi(S')$, $\sigma_S$ is the ordering of lowest index copies in $\sigma'_{S'}$.

For cross-monotonicity, consider $S' \subseteq T', \sigma_{S'} \subseteq \sigma_{T'}$ such that $S'$ is a ``partial prefix" of $T'$. Now, for any $i'\in S'$, if $i'$ is not a lowest indexed copy in $T'$, then $\chi(i',T', \sigma'_{T'})=0$, so that the condition is automatically satisfied. Let $i'$ is one of the lowest indexed copies in $T'$, then it must have been a lowest indexed copy in $S'$, since $S'$ is a subset of $T'$, and $\sigma_{S'} \subseteq \sigma_{T'}$. 
Thus,
$$
\chi(i',T',\sigma'_{T'}) = \chi(i, T, \sigma_T) \le  \chi(i, S, \sigma_S) =  \chi(i',S', \sigma'_{S'})
$$
where $S = \Pi(S'), T= \Pi(T')$, $\sigma_S, \sigma_T$ are the orderings of lowest indexed copies in $S', T'$ respectively. 
Note that the inequality in above uses cross-monotonicity of $\chi$, which is satisfied only if in addition to $S\subseteq T$, we have that $\sigma_S \subseteq \sigma_T$. That is, if the ordering of elements of $S$ is same in $\sigma_S$ and $\sigma_T$. We show that this is true given the assumption that $\sigma_{S'}, \sigma_{T'}$ respect the partial ordering $A_K, \ldots, A_1$, and $S'$ is a ``partial prefix" of $T'$. That is, $S' \subseteq A_K \cup \cdots \cup A_k$, and $T'\backslash S' \subseteq A_k \cup \cdots \cup A_1$ for some $k$. 
\addconf{\balance}
To see this, observe that the splitting was performed in a manner so that atmost one copy of any element appears in each $A_k$. So, among the newly added copies $T'\backslash S'$, any copy of an element of $S$ can occur only in $T' \cap A_{k+1}$ or later. Since $S' \subseteq A_1 \cup \cdots \cup A_k$, this means that for any element $i\in S$, the newly added copies occur only later in the ordering and they cannot alter the order of lowest indexed copies of elements of $S$. This proves that $\sigma_S \subseteq \sigma_T$.

\end{proof}

\rmconf{
\section{Proof of Lemma \ref{th:supermodular}}
\label{app:supermodular}
For any fixed $x$, the worst case expected cost is the optimal value of following linear program, where $\{\alpha_S\}_{S\subseteq V}$ represents a distribution over subsets of set $V$:

\begin{equation}\label{eq:worstLP}
\begin{array}{lrl}
\LP(f,V,\{p_i\}) =& \max_{\alpha} & \sum_{S}{\alpha_S f(x,S)}\\
& \mbox{s.t.} & \sum_{S:~i \in S}{\alpha_S} = p_i,~\forall i \in V \\
& & \sum_{S} {\alpha_S}=1 \\
& & \alpha_S \geq 0,~ \forall S \subseteq V. \\
\end{array}
\end{equation}
It is easy to verify that
$$
\alpha^* = \left\{ \begin{array}{ll}
									p_n & \mbox{ if } S=S_n\\
									(p_i - p_{i+1}) & \mbox{ if } S=S_i, 1 \le i \le n-1\\
									1-p_1				& \mbox{ if } S = \emptyset\\
									0 					& o.w.
									\end{array}
									\right.
									$$
is a feasible solution to \eqref{eq:worstLP}. Next we show that it is actually the optimal solution. The dual of linear program \eqref{eq:worstLP} is:
\begin{equation}
\label{dual-supermodular}
\begin{array}{ll} \min_{\gamma, \lambda} & \gamma + {p^T \lambda}   \\
\mbox{s.t.}  & f(S) - \sum_{i \in S}{\lambda_i}  \leq \gamma,~ \forall S.
\end{array}
\end{equation}
Consider the problem in $\lambda$ for a given value of $\gamma$. This problem is to minimize a linear function with positive coefficients ($p_i$) over the supermodular polyhedron (of supermodular function $f(S) - \gamma$). Minimizing a linear function over a supermodular polyhedron can be solved by a greedy procedure \cite{edmonds70}, with the optimal value given by $\sum_{i=1}^n p_i (f(S_i) - f(S_{i-1}))$.
Then \eqref{dual-supermodular} can be rewritten as
$$
\begin{array}{ll}
\min_{\gamma} & \gamma +  p_n f(S^n) + \displaystyle\sum_{i=1}^{n-1} (p_i - p_{i+1}) f(S^{i}) - p_1 f(\emptyset) \\
\mbox{s.t.}  & f(\emptyset) \leq \gamma. \\
\end{array}
$$
The optimal solution for above is $\gamma=f(\emptyset)$, therefore optimal value:
$$
\begin{array}{l}
p_n f(S^n) + \sum_{i=1}^{n-1} (p_i - p_{i+1}) f(S^{i}) + (1 - p_1)  f(\emptyset) \\
= \sum_S \alpha^*_S f(S)
\end{array}
$$
This proves the lemma.
}

\end{document}